\documentclass[10pt]{article}
\usepackage{mathrsfs}
\usepackage{amsmath,amsfonts}
\usepackage{amssymb,amscd,amsthm}
\usepackage[all]{xy}
\usepackage[dvips]{graphicx}
\usepackage{verbatim}
\usepackage[perpage,symbol*]{footmisc}
\usepackage[justification=centering]{caption}
\usepackage{longtable}
\usepackage{multirow}
\usepackage{slashbox}
\usepackage[figuresright]{rotating}
\usepackage{makecell}
\usepackage{rotating}
\usepackage[numbers]{natbib}
\usepackage{tikz}

\newtheorem{lemma}{\qquad\textbf{Lemma}}
\newtheorem{theorem}{\qquad\textbf{Theorem}}

\newtheorem{corollary}{\qquad\textbf{Corollary}}

\newtheorem{example}{\qquad\textbf{Example}}

\newtheorem*{t1}{\qquad\textbf{Riemann-Roch Theorem}}
\newtheorem*{t2}{\qquad\textbf{Hermitian construction}}
\newtheorem*{t3}{\qquad\textbf{Upper Hasse-Weil bound}}

\begin{document}
\baselineskip 17pt
\title{\Large\bf  Self-orthogonal and self-dual codes from maximal curves}
\author{\large  Puyin Wang \quad\quad Jinquan Luo*}\footnotetext{The authors are with School of Mathematics
and Statistics \& Hubei Key Laboratory of Mathematical Sciences, Central China Normal University, Wuhan China 430079.\\
 E-mail: p.wang98@qq.com(P.Wang), luojinquan@ccnu.edu.cn(J.Luo)}
\date{}
\maketitle
{\bf Abstract}:
    In the field of algebraic geometric codes(AG codes), the characterization of dual codes has long been a challenging problem which relies on differentials. In
    this paper, we provide some descriptions for certain differentials utilizing algebraic structure of finite fields and geometric properties of
    algebraic curves. Moreover, we construct self-orthogonal and self-dual codes with parameters $[n, k, d]_{q^2}$ satisfying $k + d$ is close to
    $n$. Additionally, quantum codes with large minimum distance are also constructed.

{\bf Key words}: Algebraic geometry code,  Hermitian self-orthogonal code, Euclidian self-orthogonal code, Euclidian self-dual code, Quantum code.

\section{Introduction}
\qquad Let $q$ be a power of a prime $p$, and let $\mathbb{F}_q$ denote the finite field with $q$ elements. An $[n,k,d]$ linear code $C$
over $\mathbb{F}_q$ is a $k$-dimensional subspace of $\mathbb{F}_q^n$ with minimum distance $d.$  For a linear code $C$, the Euclidean dual $C^\perp$ is defined as
$$C^\perp=\left\{(c_1^\prime,c_2^\prime,\ldots,c_n^\prime)\in \mathbb{F}_q^n\mid\sum_{i=1}^nc_i'c_i=0,\:\forall\:(c_1,c_2,\ldots,c_n)\in C\right\}.$$
The code $C$ is Euclidean self-orthogonal if it satisfies $C\subseteq C^\perp$. It is Euclidean self-dual if  $C=C^\perp$.

Similarly, for a linear code $C$ over $\mathbb{F}_{q^2}$,  the Hermitian dual $C^{\perp_H}$ is defined as
$$C^{\perp_H}=\left\{(c_1^\prime,c_2^\prime,\ldots,c_n^\prime)\in \mathbb{F}_{q^2}^n\mid\sum_{i=1}^n{c_i^\prime}^qc_i=0,\:\forall\:(c_1,c_2,\ldots,c_n)\in C\right\}.$$
The code $C$ is  Hermitian self-orthogonal if it satisfies $C\subseteq C^{\perp_H}$. It is  Hermitian self-dual if $C=C^{\perp_H}$.

Euclidean self-orthogonal codes not only possess special structures but also have a wide range of applications. For instance, they can be applied to construct quantum codes \cite{CRSS}. Moreover,
self-dual code, as a special type of self-orthogonal code, holds many  theoretical significances. They can be applied to construct some $t$-designs, such
as \cite{AM} and \cite{GB}. Besides, they are closely connected to lattices and modular forms \cite{BBH}, \cite{H}, \cite{Sv}.  Furthermore, self-dual codes have applications in secret sharing. In \cite{CDGU}, linear secret sharing schemes with specific access structure are constructed from self-dual codes.

In recent years, significant progress has been made in the construction of Euclidean self-orthogonal codes. For instance, \cite{ZKL} explores their construction through cyclic codes. In \cite{SWKS}, Euclidean self-orthogonal codes are constructed via non-unital rings, and \cite{ZLD} presents the construction of several Euclidean self-orthogonal codes in a generic way.

Euclidean self-dual codes have received great attention in recent years. In \cite{JX}, a generic criteria is proposed to verify a generalized Reed-Solomon(GRS) code being self-dual. Later, it is generalized to extended GRS case in \cite{HY}. In \cite{ZF}, a method to construct longer Euclidean MDS self-dual codes from shorter length is proposed, which is generalized in \cite{XFL}.  Additionally, several classes of MDS self-dual codes are constructed via multiplicative subgroups or additive subgroups in finite fields  (\cite{FZXF}, \cite{FLLL},  \cite{SYS}).  In particular, a large class of MDS self-dual codes are derived in \cite{WLZ} via union of subsets over two multiplicative subgroups in finite fields.  \cite{SYS}. Also,  \cite{FL2} focuses on the construction of self-dual 2-quasi negacyclic codes.

Self-orthogonal codes with respect to the Hermitian inner product play a crucial role in the construction of quantum codes (\cite{AK}, \cite{KKKS}). It has been extensively
studied in recent years (e.g., \cite{GLLS}). Notably, quantum MDS codes have been constructed using GRS codes, with the Hermitian construction proposed in
\cite{AK} in recent years, such as \cite{JLLX}, \cite{LLGS}, \cite{WT}.

\begin{t2}
If there exists an $[n,k,d]_{q^2}$ linear code $C$ with $C\subseteq C^{\perp_H}$, then there exist an $[[n,n-2k,\geqslant d^{\perp}]]_q$ quantum code where $d^{\perp}$ is the minimum distance of $C^{\perp_H}$.\hfill$\square$
\end{t2}

Using algebraic curves over finite fields, Goppa introduced an influential class of error-correcting codes known as algebraic-geometric (AG) codes \cite{G}.
 By leveraging the geometric properties of algebraic curves, AG code achieves higher code rate while preserving larger
minimum distance for the same code length. This allows them to exceed the Gilbert-Varshamov (GV) bound, a threshold that classical codes often struggle to reach. By effectively
balancing code rate and minimum distance, AG codes deliver superior error-correction performance for a given code length. These properties can be applied in several communication scenarios such as optimal locally repairable codes \cite{MX}.

By leveraging the geometric properties of algebraic curves, such as choosing appropriate set of rational points, we can get a series of self-orthogonal algebraic geometry codes. For instance, certain Hermitian self-orthogonal codes with strong performance are derived from algebraic geometry codes \cite{JLLX}. In fact, by employing
appropriate differential, we can characterize  algebraic geometry codes are Euclidean or Hermitian self-orthogonal. Related work in recent years can be found in
\cite{HMMF},\cite{J}.

In this paper, we construct some self-orthogonal codes via two kinds of maximal curves. The parameters of these codes are also determined/estimated. The paper is organized as follows. In section $2$,
we introduce some preliminary knowledge on AG codes which will be useful for the remaining paper. In section $3$, we construct two classes of self-orthogonal AG codes from two
classes maximal curves. And we give some quantum codes with large minimal distances. In section $4$, we will make a conclusion and some further problems will be proposed.

\section{Preliminary}

\qquad In this section, we introduce some basic notations and auxiliary results related to algebraic curves and algebraic geometry codes. Further details are available in \cite{S}.

\subsection{Places, divisors and Riemann-Roch space}

\qquad Let $\mathbb{F}_q$ be the finite field with $q$ elements where $q$ is a power of prime number $p$. Let $\chi$ be a be a smooth, projective, absolutely irreducible curve of genus $g$ defined over
$\mathbb{F}_q$. Denote by $\mathbb{F}_q(\chi)$ the function field of $\chi$. Let $\mathbb{P}_{\mathbb{F}_q(\chi)}$ be the set of places of $\mathbb{F}_q(\chi)$. Let $v_P$ be the
normalized discrete valuation corresponding to a point $P$ of $\chi$. Denote by $\#S$ the number of elements in a finite set $S$. Then $\chi$ is said to be maximal if it attains the
upper Hasse-Weil bound.

\begin{t3}
    $\#\mathbb{P}_{\mathbb{F}_q(\chi)}\leqslant q+1+2g\sqrt{q}$.\hfill$\square$
\end{t3}

For a divisor $D=\sum\limits_{P\in\mathbb{P}_{\mathbb{F}_q(\chi)}}n_PP$ where $n_P$ is integer and almost all $n_P$ are $0$, the support of $D$ is ${\rm
supp}(D)=\{P\in\mathbb{P}_{\mathbb{F}_q(\chi)}\mid n_P\neq0\}$. Define the degree of $D$ to be
$\deg(D)=\sum\limits_{P\in {\rm supp}(D)}n_P$.

Let $\Omega$ be the differential space of $\chi$. We define
$$\Omega(G)=\{\omega\in\Omega\setminus\{0\}\mid\mathrm{div}(\omega)\geqslant G\},$$
where $(\omega)$ is the canonical divisor corresponding to $\omega$. In fact, all canonical divisors have degree $2g-2$. Define the principal divisor of $x \in \mathbb{F}_q(\chi)$ as
following:
$$(x):=\sum\limits_{P \in {\mathbb{P}_F}} {{v_P}(x)P}.$$
The Riemann-Roch space associated to $D$ is
$$\mathscr{L}(D)=\{x\in \mathbb{F}_q(\chi)\mid(x)+D\geqslant 0\}\cup\{0\}.$$
We denote the dimension of $\mathscr{L}(D)$ by $l(D)$.

The connection between the above concepts is elaborated by the following result.

\begin{t1}
    Let $W$ be a canonical divisor of $\mathbb{F}_q(\chi)$. Then for each divisor $G$ of $\mathbb{F}_q(\chi)$, we have
    $$l(G)=\deg(G)+1-g+l(W-G).$$
Especially, we have $l(G)=\deg(G)+1-g$ if $\deg(G)> 2g-2$.\hfill$\square$
\end{t1}

\subsection{AG codes}

\qquad  Suppose $P_{1},\ldots,P_{n}$ are pairwise distinct rational places of $\mathbb{F}_q(\chi)$. $D=P_{1}+\ldots+P_{n}$ and $G$ is a divisor of $\mathbb{F}_q(\chi)/\mathbb{F}_q$
such that ${\rm supp}(G)\cap{\rm supp}(D)=\emptyset$.

Consider the evaluation map $\mathrm{ev}_{D}:\mathscr{L}(G)\to \mathbb{F}_q^{n}$ given by
$$ev_{D}(f)=(f(P_{1}),\ldots,f(P_{n}))\in \mathbb{F}_q^{n}$$
under the isomorphism between the residue fields of rational places and $\mathbb{F}_q$. The image of $\mathscr{L}(G)$ under $ev_{D}$ is the algebraic geometric code (or AG code
shortly) $C_\mathscr{L}(D,G)$. The Euclidean dual code of $C_\mathscr{L}(D,G)$ is
$$C_{\Omega}(D,G)=\{(\mathrm{res}_{P_1}(\omega),\ldots,\mathrm{res}_{P_n}(\omega))\mid\omega\in\Omega_{F}(G-D)\}.$$

\begin{theorem}{\rm(\cite{S} Theorems 2.2.2  and  2.2.7)}\label{AG}
    $C_{\mathscr{L}}(D,G)$ is an $[n,k,d]$ code with parameters
    \begin{equation}\notag
        k=l(G)-l(G-D), d\geqslant n-\deg(G).
    \end{equation}
   Its dual code $C_{\Omega}(D,G)$ is an $[n,k^{\prime},d^{\prime}]$ code with parameters
    \begin{equation}
        k^{\prime}=n-l(G)+l(G-D), d^{\prime}\geqslant\deg(G)-(2g-2).\tag*{\qed}
    \end{equation}
\end{theorem}

The following result in \cite{S} shows the connection between $C_\mathscr{L}(D,G)$ and $C_\Omega(D,G)$.
\begin{theorem}{\rm(\cite{S} Proposition 8.1.2)}\label{D}
    For $D=\sum\limits_{i = 1}^n {{P_i}}$, let $\eta$ be a differential such that $v_{P_i}(\eta)=-1$ and $res_{P_i}(\eta)=1$ for $i=1,2,\cdots,n$. Then
    \begin{equation}
    C_\Omega(D,G)=C_\mathscr{L}(D,D-G+(\eta))\tag*{\qed}
    \end{equation}
\end{theorem}

Besides, the following result in \cite{HY} is useful in our construction.
\begin{lemma}\label{M}
    Suppose $m\mid q-1$ and $\alpha\in \mathbb{F}_q$ is an $m$-th primitive root, Then for $1\leqslant i\leqslant m$,
    \begin{equation}
    \prod\limits_{1\leqslant i\leqslant m,j \ne i} {({\alpha ^i} - {\alpha ^j})}  = m{\alpha ^{ - i}}.\tag*{\qed}
    \end{equation}
\end{lemma}

\section{Self-orthogonal and self-dual codes}
In this section, for an irreducible curve $\chi$, we always regard $\mathbb{F}_{q^2}(\chi)/\mathbb{F}_{q^2}$ as  algebraic extension of $\mathbb{F}_{q^2}(x)/\mathbb{F}_{q^2}$.
\subsection{Codes from $y^q+y=x^m$}

\qquad   For $m\mid q+1$ and $p\mid m-1$, denote $\mathbb{F}_{q^2}(\chi)$ the function field of $\chi$ over $\mathbb{F}_{q^2}$, where $\chi$
$$\chi:y^q+y=x^m.$$
Then the genus of $\chi$ is $g=\frac{1}{2}(m-1)(q-1)$. Special case of this curve with $q$ being odd power of $2$ and $m=3$ has been studied in \cite{J}, which is employed to construct quantum codes. Here we will investigate more general case.

Let $n=q(m(q-1)+1)$. Let $\{P_1,P_2,\cdots,P_{m(q-1)}\}$ be the set of rational places of $\mathbb{F}_{q^2}(x)$ corresponding to $m(q-1)$-th units roots in $\mathbb{F}_{q^2}$. Let $P_0$ be
the rational place corresponding to $0$ and let $P_\infty$ be the infinite place of $\mathbb{F}_{q^2}(x)$. In the extension $\mathbb{F}_{q^2}(\chi)/\mathbb{F}_{q^2}(x)$,  finite places $P_i$ split completely and $P_\infty$ is totally ramified. Let $Q_{i ,j}$ be all places of $\mathbb{F}_{q^2}(\chi)$ that lie over $P_i$ for $1\leq i\leq m(q-1)$ and let $Q_{\infty}$ be the infinite place lying above $P_{\infty}$. This means the number of rational places of $\mathbb{F}_{q^2}(\chi)/\mathbb{F}_q$ is at least
$N=mq(q-1)+1=2gq+q^2+1$, which meets the upper Hasse-Weil bound. Hence $\chi$ is a maximal curve.
 Denote by $D=\sum\limits_{i=0}^{m(q-1)}\sum\limits_{j=1}^q {{Q_{i,j}}}$.

\begin{lemma}\label{res}
    For $r>0$, $C_\mathscr{L}(D,rQ_\infty)^\perp=C_\mathscr{L}(D,(n+2g-2-r)Q_\infty)$.\hfill$\square$
\end{lemma}

\begin{proof}

    Consider $\eta=\frac{-dx}{x(x^{m(q-1)}-1)}$. Then it is easy to see $v_{Q_{i,j}}(\eta)=-1$ for all $i,j$. According to Lemma \ref{M}, $res_{Q_{i,j}}(\eta)=\frac{-1}{m(q-1)}=1$ for all $i,j$.
    On the other hand,
    $$
    \begin{aligned}
    \nonumber
        (\eta) &= q(m(q-1)+1)Q_\infty-D+(dx)           \\
               &= (n+2g-2)Q_\infty-D,
    \end{aligned}
    $$
    which implies $D-G+(\eta)=(n+2g-2-r)Q_\infty$. Then the conclusion follows.
\end{proof}

\begin{lemma}\label{dim} {\rm(\cite{S} Proposition 6.4.1)}
    Consider a function field $F=K(x, y)$ with
    $$y^q+\mu y=f(x)\in K[x],$$
    where $q= p^{s}> 1$ is a power of $p$ and $0\neq \mu \in K$. Assume that $\deg(f):= m > 0$ is prime to $p$, and that all roots of $T^{q}+ \mu T= 0$ are in $K$. Let $r\geq 0$.
    Then the elements $x^{i}y^{j}$ with $0\leq i$, $0\leq j\leq q-1$, $qi+mj\leq r$ form a basis of the space $\mathscr{L}(rQ_{\infty})$ over $K$.\hfill$\square$
\end{lemma}

Now we can construct Euclidean self-orthogonal and self-dual codes with good parameters.
\begin{theorem}\label{mc}
    \begin{itemize}
        \item[{\rm(1).}]For $mq-m-q\leqslant r\leqslant\lfloor\frac{m(q^2-1)-1}{2}\rfloor$, the code $C_\mathscr{L}(D,rQ_\infty)$ is a Euclidean self-orthogonal code with parameters $[n,k_0,d_0]_{q^2}$ where \[n=mq^2-mq+q,\quad k_0\geqslant r-\frac{1}{2}(m-1)(q-1)+1, \quad d_0\geqslant n-r.\]
        \item[{\rm(2).}]In particular, there exist Euclidean self-dual codes with parameters $[mq^2-mq+q,\frac{mq^2-mq+q}{2},\geqslant \frac{1}{2}(mq^2+m+1)-(m-1)q]_{q^2}$ when $q$ is even.\hfill$\square$
    \end{itemize}
\end{theorem}

\begin{proof}
\begin{itemize}
  \item[(1).] Since $2r\leqslant n+2g-2=m(q^2-1)-1$, the inclusion $C_\mathscr{L}(D,rQ_\infty)\subseteq C_\mathscr{L}(D,(n+2g-2-r)Q_\infty)$ holds. In this case, the code $C_\mathscr{L}(D,rQ_\infty)$ has length $n$. Since $2g-2<r<\mathrm{deg}\, D=mq^2-mq+q$, by Riemann-Roch Theorem its dimension($W$ is a canonical divisor)
       \[k_0=l(rQ_{\infty})-l(rQ_{\infty}-D)=l(rQ_{\infty})= r-g+1+l(W-rQ_{\infty})\geqslant r-\frac{1}{2}(m-1)(q-1)+1.\]
       Moreover, for $r>2g-2=mq-m-q-1$, the equality holds, that is, $k_0=r-\frac{1}{2}(m-1)(q-1)+1$.
       Its minimal distance is at least $n-\mathrm{deg}\, (rQ_{\infty})=n-r$.
  \item[(2).] Consider $2r={m(q^2-1)-1}$. Then $C_\mathscr{L}(D,rQ_\infty)= C_\mathscr{L}(D,(n+2g-2-r)Q_\infty)=C_\mathscr{L}(D,rQ_\infty)^{\perp}$. Then $C_\mathscr{L}(D,rQ_\infty)$ is self-dual with desired parameters.
\end{itemize}

\end{proof}

In this way, we obtain $[n, \frac{n}{2}, \geqslant \frac{1}{2}(mq^2+m+1)-(m-1)q]_{q^2}$ self-dual codes with minimal distance close to $n/2$. Moreover, some Hermitian self-orthogonal codes can be derived.

\begin{theorem}\label{coro}
    For $mq-m-q\leqslant r\leqslant m(q-1)-1$,
    \begin{itemize}
    \item[{\rm(1).}] $C_\mathscr{L}(D,rQ_\infty)$ is Hermitian self-orthogonal with parameters $[mq^2-mq+q,k_0,d_0]_{q^2}$ with
    \[k_0=r-\frac{1}{2}(m-1)(q-1)+1, \quad d_0 \geqslant n-r.\]
    \item[{\rm(2).}] there exists  $q$-ary  $[[mq^2-mq+q,k_1, d_1]]_q$ quantum code with
      \[k_1=mq^2-m-2r-1,\quad d_1\geqslant r-mq+m+q+1.\]
    \end{itemize}

\end{theorem}

\begin{proof}
\begin{itemize}
  \item[(1).] Note that $f \in \mathscr{L}(G)$ deduces $f^q \in \mathscr{L}(qG)$. By Lemma \ref{dim}, $\mathscr{L}(rQ_\infty)$ has a basis of the form $x^iy^j$. Therefore  $C_\mathscr{L}(D,rQ_\infty)^q\subseteq C_\mathscr{L}(D,qrQ_\infty)$.
The condition $r\leqslant m(q-1)-1$ implies $qr\leq n+2g-2-r$. Hence
\[C_\mathscr{L}(D,rQ_\infty)^q \subseteq C_\mathscr{L}(D,qrQ_\infty)\subseteq C_\mathscr{L}(D,(n+2g-2-r)Q_\infty)=C_\mathscr{L}(D,rQ_\infty)^{\perp}.\]
Hence $C_\mathscr{L}(D,rQ_\infty)$ is Hermitian self-orthogonal with desired parameters.
  \item[(2).] Applying Hermitian construction to $C_\mathscr{L}(D,rQ_\infty)$ in (1) derives the quantum code.
\end{itemize}
\end{proof}

\begin{example}
    Let $q=27, m=7$ and  $r=181$.  By Theorem \ref{coro} we obtain $[4941,104,\geqslant4760]_{729}$ code which is Hermitian self-orthogonal. Using Hermitian construction, we obtain quantum code with parameters $[[4941,4733,\geqslant27]]_{27}$.
\end{example}

The following result for $q$ being odd power of $2$ and $m=3$ has been investigated in \cite{J}.
\begin{corollary}\label{coro_m}
    For odd $m$ and $q$ is power of $2$, there exists  $q$-ary  $[[mq^2-mq+q,k_1, d_1]]_q$ quantum code with $k_1=mq^2-m-2r-1,\quad d_1\geqslant r-mq+m+q+1$. Especially, let $m=3$ and $q$ is odd power of $2$, we have quantum code with parameter $[[3q^2-2q,3q^2-4-2r,\geqslant r+4-2q]]_q$.
\end{corollary}
\begin{proof}
  For $q=p^l$ with $p=2$, $l$ odd,  note that $3\mid 2^l+1$ and $2\mid m-1$. The the conclusion follows from  Theorem \ref{coro}.
\end{proof}

\begin{example}
For $q=8, m=3, r=20$, the corresponding code $C_\mathscr{L}(D,rQ_\infty)$ is Hermitian self-orthogonal with parameters $[176,14,\geqslant156]_{64}$. By Hermitian construction we obtain quantum code with parameters $[[176,148,\geqslant8]]_8$.
\end{example}

\subsection{Codes from Hermitian curves}

\qquad Denote by $\mathbb{F}_{q^2}(\chi)$ the Hermitian function field of $\chi$ over $\mathbb{F}_{q^2}$ with
$$\chi:y^q+y=x^{q+1}.$$
The genus of $\chi$ is $g=\frac{q(q-1)}{2}$.

\subsubsection{From multiplicative groups}

\qquad For $s\mid q^2-1$, $p\mid s+1$ and $n=q(s+1)$, let $\{P_1,P_2,\cdots,P_s\}$ be the set of rational places of $\mathbb{F}_{q^2}(x)$ corresponding to $s$-th unit roots in
$\mathbb{F}_{q^2}$. Let $P_0$ be the rational places corresponding to $0$ and let $P_\infty$ be the infinite places of $\mathbb{F}_{q^2}(x)$. Then these finite places $P_i$ split completely and $P_\infty$ is totally ramified.  Let $Q_{i ,j}$ be all places of
$\mathbb{F}_{q^2}(\chi)$ that lie above $P_i$. Let $Q_\infty$ be the unique place in $\mathbb{F}_{q^2}(\chi)$ lying above $P_{\infty}$.  Denote by $D=\sum\limits_{i=0}^{s}\sum\limits_{j=1}^q{{Q_{i,j}}}$.

\begin{lemma}
    For $r>0$ , $C_\mathscr{L}(D,rQ_\infty)^\perp=C_\mathscr{L}(D,(n+2g-2-r)Q_\infty)$.
\end{lemma}

\begin{proof}
    Consider $\eta=\frac{-dx}{x(x^s-1)}$. It is easy to see $v_{Q_i}(\eta)=-1$. According to Lemma \ref{M},
    $$
    res_{Q_i}(\eta)=
        \left\{
            \begin{aligned}
            \nonumber
                &-\frac{1}{s}=1        &   i \neq 0,            \\
                &1                     &   i=0.                 \\
            \end{aligned}
        \right.
    $$
    Then the remainder is  similar to the proof of Theorem \ref{res}.
\end{proof}

\begin{theorem}\label{q(q+1)}
    \begin{itemize}
        \item[{\rm(1).}] For $q^2-q-1\leq r\leqslant\lfloor\frac{1}{2}q(q+s)-1\rfloor$, the code $C_\mathscr{L}(D,rQ_\infty)$ is Euclidean self-orthogonal with parameters $[q(s+1),k_0,d_0]_{q^2}$ where $k_0=r-\frac{1}{2}q(q-1)+1$ and  $d_0\geqslant q(s+1)-r$.
        \item[{\rm(2).}] In particular, there exist Euclidean self-dual codes with parameters \[\left[q(s+1),\frac{q(s+1)}{2}, \geqslant\frac{1}{2}q(s-q+2)+1\right]_{q^2}\] when $q$ is even or  $s$ is  odd.
    \end{itemize}
\end{theorem}
\begin{proof}
    \begin{itemize}
        \item[(1).] The proof is similar to that of Theorem \ref{mc}. Denote by $n=q(s+1)$. Since $2r\leqslant n+2g-2$, the inclusion $C_\mathscr{L}(D,rQ_\infty)\subseteq C_\mathscr{L}(D,(n+2g-2-r)Q_\infty)$ holds. In this case, the code $C_\mathscr{L}(D,rQ_\infty)$ has length $n$, minimal distance at least $n-\mathrm{deg}\, (rQ_\infty)=n-r$. The dimension $k_0$ can be derived from  Theorem \ref{AG} and Riemann-Roch Theorem.
        \item[(2). ] Consider $2r={q(k+1)-1}$. Then $C_\mathscr{L}(D,rQ_\infty)= C_\mathscr{L}(D,(n+2g-2-r)Q_\infty)=C_\mathscr{L}(D,rQ_\infty)^{\perp}$. Then $C_\mathscr{L}(D,rQ_\infty)$ is self-dual with desired parameters.
    \end{itemize}
\end{proof}

\subsubsection{From additive groups}

\qquad Suppose $q=p^t$. For $k\leqslant2t$, let $V_k$ be a $k$-dimension $\mathbb{F}_p$-subspace of $\mathbb{F}_{q^2}$.
Let $n=q\cdot p^k$ and let $\{P_1,P_2,\cdots,P_{p^k}\}$ be the set of rational places corresponding to the elements in $V_k$. Let $P_0$ be the rational places corresponding to $0$ and let
$P_\infty$ be the infinite places of $\mathbb{F}_{q^2}(x)$.  Similarly as the pervious case, all finite places $P_i$ split completely in $\mathbb{F}_{q^2}(\chi)$ and $P_\infty$ is totally ramified in $\mathbb{F}_{q^2}(\chi)$. Let $Q_{i ,j}$ ($Q_{\infty}$ resp.) be all places of $\mathbb{F}_{q^2}(\chi)$ that lie over $P_i$ ($P_{\infty}$ resp.). Denote by $D=\sum\limits_{i
= 1}^{p^k}\sum\limits_{j=1}^q {{Q_{i,j}}}$.

\begin{lemma}
    For $r>0$, $C_\mathscr{L}(D,rQ_\infty)^\perp=C_\mathscr{L}(D,(n+2g-2-r)Q_\infty)$.\hfill$\square$
\end{lemma}

\begin{proof}
    Consider $h(x)=\prod\limits_{\alpha  \in V_k} {(x - \alpha )}$. Then  $res_{Q_{i,j}}(h(x))$ is constant for all $Q_{i,j}$ that lie over $P_i$ with $1\leqslant i\leqslant n$ since
    $V$ is stable under translation by $\alpha\in V_k$. Denote by $\lambda=res_{Q_{i,j}}(h(x))$. Let $\eta=\frac{\lambda dx}{h(x)}$. Then
    $(\eta)=(n+2g-2)Q_{\infty}-D$. By Theorem \ref{D},
    \[C_\mathscr{L}(D,rQ_\infty)^\perp=C_\Omega(D,rQ_\infty)=C_\mathscr{L}(D,D-rQ_\infty+(\eta))=C_\mathscr{L}(D,(n+2g-2-r)Q_\infty).\]
\end{proof}

\begin{theorem}
    \begin{itemize}
        \item[{\rm(1).}] For $q^2-q-1\leq r\leqslant\lfloor\frac{q(p^k+q-1)}{2}-1\rfloor$, $C_\mathscr{L}(D,rQ_\infty)$ are Euclidean self-orthogonal codes whose parameters are $[q p^k,k_0,d_0]_{q^2}$ with
            \[k_0=r-\frac{1}{2}q(q-1)+1, \quad d_0\geqslant  p^kq-r.\]
        \item[{\rm(2).}] In particular, there exist Euclidean self-dual codes with parameters $[p^kq,\frac{p^kq}{2},\geqslant\frac{ p^kq}{2}-\frac{q(q+1)}{2}+1]_{q^2}$ when $q$ is even.\hfill$\square$
    \end{itemize}
\end{theorem}

\begin{proof}
    The proof is similar to that of Theorem \ref{q(q+1)} by replacing $n$ with $q\cdot p^k$.
\end{proof}

Summarizing it up, we obtain self-orthogonal codes and self-dual codes with $k_0 + d_0$ being close to the code length $n$ if $k$ is close to $2t$. Therefore, these self-orthogonal and self-dual codes are not far from the Singleton bound while their code lengths are much longer than MDS codes.

\section{Conclusion and further study}

\qquad In this paper, we construct self-orthogonal codes over function fields which are Artin-Schreier extensions of $\mathbb{F}_{q^2}(x)$. These codes have both good parameters and Euclidean/Hermitian self-orthogonal(self-dual) properties.  Moreover, Kummer extension has abundant structure as well as Artin-Schreier extension. Therefore, it may be possible to construct self-orthogonal codes with better parameters from Kummer extension as well.

\newpage

\bibliographystyle{IEEEtran}
\bibliography{IEEEexample}

\begin{thebibliography}{1}

        \bibitem{AK} A. Ashikhmin and E. Knill,
        ``Nonbinary quantum stabilizer codes," \textit{IEEE Trans. Inform. Theory,} vol. 47, no. 7, pp. 3065-3072, Nov. 2001.

        \bibitem{AM} E. F. Assmus, Jr. and H. F. Mattson, Jr,
        ``New 5-designs," \textit{J. Combinatorial Theory,} vol. 6, no. 2, pp. 122-151, Mar. 1969.

        \bibitem{GB} C. Bachoc and P. Gaborit,
        ``Designs and self-dual codes with long shadows," \textit{J. Combin. Theory Ser. A,}  vol. 105, no. 1, pp. 15-34, Jan. 2004.

        \bibitem{BBH} S. Bouyuklieva, I. Bouyukliev, and M. Harada,
        ``Some extremal self-dual codes and nunimodular lattices in dimension $40$," \textit{Finite Fields Appl.,} vol. 21, pp. 67-83, May 2013.

        \bibitem{CRSS} A. R. Calderbank, E. M. Rains, P. W. Shor, and N. J. A. Sloane,
        ``Quantum error correction via codes over $GF(4)$," \textit{IEEE Trans. Inform. Theory,}  vol. 44, no. 4, pp. 1369-1387, Jul. 1998.

        \bibitem{CDGU} R. Cramer, V. Daza, I. Gracia, J. J. Urroz,  C. Padr$\acute{o}$,
        ``On codes, matroids and secure multi-party computation from linear secret sharing schemes," \textit{Adv. Cryptology-CRYPTO2005 (Lecture Notes in Computer Science).} Berlin, Germany: Springer-Verlag, 2005, vol. 3621, pp. 327-343.

        \bibitem{FL2} Y. Fan and Y. Leng,
        ``Self-dual $2$-quasi negacyclic codes over finite fields," \textit{Finite Fields Appl.,} vol. 101, 2025. Art. no. 102541.

        \bibitem{FZXF} W. Fang, J. Zhang, S.-T. Xia, and F.-W. Fu,
        ``New constructions of self-dual generalized Reed-Solomon codes," \textit{Cryptogr. Commun.,} vol. 14, no. 3, pp. 677-690, May 2022.

        \bibitem{FLLL} X. Fang, K. Lebed, H. Liu, and J. Luo,
        ``New MDS self-dual codes over finite fields of odd characteristic," \textit{Des. Codes Cryptogr.,} vol. 88, no. 6, pp. 1127-1138, Jun. 2020.

        \bibitem{G} V. D. Goppa,
        ``Codes that are associated with divisors," \textit{Problemy Pereda$\check{c}$i Informacii,} vol. 13, no. 1, pp. 33-39, 1977.

        \bibitem{GLLS} G. Guo, R. Li, Y. Liu, and H. Song,
        ``Duality of generalized twisted Reed-Solomon codes and Hermitian self-dual MDS or NMDS codes," \textit{Cryptogr. Commun.,} vol. 15, no. 2, pp. 383-395, Mar. 2023.

        \bibitem{HZ} D. Han and H. Zhang,
        ``Explicit constructions of NMDS self-dual codes," \textit{Des. Codes Cryptogr.,} vol. 92, no. 11, pp. 3573-3585, Nov. 2024.

        \bibitem{H} M. Harada,
        ``On the existence of frames of the Niemeier lattices and self-dual codes over $\mathbb{F}_p$," \textit{J.Algebra,} vol. 321, no. 8, pp. 2345-2352, Apr. 2009.

        \bibitem{HMMF} F. Hernando, G. McGuire, F. Monserrat, and J. J. Moyano-Fern$\acute{a}$ndez,
        ``Quantum codes from a new construction of self-orthogonal algebraic geometry codes," \textit{Quantum Inf. Process.,} vol. 19, no. 4, p. 117, Apr. 2020.

        \bibitem{J} L. Jin,
        ``Quantum stabilizer codes from maximal curves," \textit{IEEE Trans. Inform. Theory,} vol. 60, no. 1, pp. 313-316, Jan. 2014.

        \bibitem{JLLX} L. Jin, S. Ling, J. Luo, and C. Xing,
        ``Application of classical Hermitian selforthogonal MDS codes to quantum MDS codes," \textit{IEEE Trans. Inform. Theory,} vol. 56, no. 9, pp. 4735-4740, Sep. 2010.

        \bibitem{JX} L. Jin and C. Xing,
        ``New MDS self-dual codes from generalized Reed-Solomon codes," \textit{IEEE Trans. Inform. Theory,} vol. 63, no. 3, pp. 1434-1438, Mar. 2017.

        \bibitem{KKKS} A. Ketkar, A. Klappenecker, S. Kumar, and P. K. Sarvepalli,
        ``Nonbinary stabilizer codes over finite fields," \textit{IEEE Trans. Inform. Theory,} vol. 52, no. 11, pp. 4892-4914, Nov. 2006.

        \bibitem{LLGS} Z. Li, , R. Li, C. Guan , H. Song,
        ``On the construction of hermitian self-orthogonal codes over $\mathbb{F}_9$ and their application," \textit{Inter. J. Theor. Physics,} vol. 63, no. 9, Sep. 2024. Art. no. 221.

        \bibitem{MX} L. Ma and C. Xing,
        ``Constructive asymptotic bounds of locally repairable codes via function fields," \textit{IEEE Trans. Inform. Theory,} vol. 66, no. 9, pp. 5395-5403, Sep. 2020.

        \bibitem{SWKS} M. Shi, S. Wang, J.-L. Kim, and P. Sol$\acute{e}$,
        ``Self-orthogonal codes over a non-unital ring and combinatorial matrices," \textit{Des. Codes Cryptogr.,} vol. 91, no. 2, pp. 677-689, Feb. 2023.

        \bibitem{ST} H. Stichtenoth,
        ``Transitive and self-dual codes attaining the Tsfasman-Vladut-Zink bound," \textit{IEEE Trans. Inform. Theory,}  vol. 52, no. 5, pp. 2218-2224, May 2006.

        \bibitem{S} H. Stichtenoth,
        \textit{Algebraic function fields and codes}(Graduate Texts in Mathematics), vol. 254. Berlin, Germany: Springer, 2009.

        \bibitem{SYS} J. Sui, Q. Yue, and F. Sun,
        ``New constructions of self-dual codes via twisted generalized Reed-Solomon codes," \textit{Cryptogr. Commun.,} vol. 15, no. 5, pp. 959-978, Sep. 2023.

        \bibitem{Sv} F.-W. Sun and H. C. A. van Tilborg,
        ``The Leech lattice, the octacode, and decoding algorithms," \textit{IEEE Trans. Inform. Theory,} vol. 41, no. 4, pp. 1097-1106, Jul. 1995.

        \bibitem{WLZ} R. Wan, Y. Li, and S. Zhu,
        ``New MDS self-dual codes over finite field $\mathbb{F}_{r2}$," \textit{IEEE Trans. Inform. Theory,} vol. 69, no. 8, pp. 5009-5016, Aug. 2023.

        \bibitem{WT} G. Wang and C. Tang,
        ``Some quantum MDS codes from GRS codes," \textit{Linear Multilinear Algebra,} vol. 72, no. 9, pp. 1418-1430, Jun. 2024.

        \bibitem{XFL} D. Xie, X. Fang, and J. Luo,
        ``Construction of long MDS self-dual codes from short codes," \textit{Finite Fields Appl.,} vol. 72, Jun. 2021. Art. no. 101813.

        \bibitem{HY} H. Yan,
        ``A note on the constructions of MDS self-dual codes," \textit{Cryptogr. Commun.,} vol. 11, no. 2, pp. 259-268, Mar. 2019.

        \bibitem{ZF} A. Zhang and K. Feng,
        ``A unified approach to construct MDS self-dual codes via Reed-Solomon codes," \textit{IEEE Trans. Inform. Theory,} vol. 66, no. 6, pp. 3650-3656, Jun. 2020.

        \bibitem{ZKL} J. Zhang, X. Kai, and P. Li,
        ``Self-orthogonal cyclic codes with good parameters," \textit{Finite Fields Appl.,} vol. 101, 2025. Art. no. 102534.

        \bibitem{ZLD} Z. Zhou, X. Li, C. Tang, and C. Ding,
        ``Binary LCD codes and selforthogonal codes from a generic construction," \textit{IEEE Trans. Inform. Theory,} vol. 65, no. 1, pp. 16-27, Jan. 2019.
    \end{thebibliography}

\end{document}